%% file: root.tex
\documentclass[letterpaper, 10 pt, conference]{ieeeconf}
\usepackage{cite}
\usepackage{amsmath,amssymb,amsfonts}
\usepackage{lipsum}

\usepackage{algorithmic}
\usepackage{graphicx}
\usepackage[linesnumbered,ruled,vlined]{algorithm2e}
\SetKwInput{KwInput}{Input}                
\SetKwInput{KwOutput}{Output}              
\SetKwInOut{Parameter}{parameter}
\usepackage{alphabeta} 
\usepackage[utf8]{inputenc}
\usepackage{commath} 
\usepackage{nccmath} 
\usepackage{pifont}
\usepackage{optidef}
\usepackage{array}
\usepackage{makecell}
\usepackage[arrowdel]{physics}

\usepackage[position=bottom,caption=false]{subfig}
\usepackage{multirow}                		
\usepackage{colortbl} 						

\usepackage{tabularx}
\usepackage{siunitx}						
\usepackage{nameref}
\let\labelindent\relax
\usepackage{enumitem}
\setitemize[0]{leftmargin=*}
\usepackage{hyperref}
\usepackage{pifont}
\usepackage{marvosym}

\include{preamble_pgfplots} 
\newtheorem{proposition}{Proposition}
\newtheorem{definition}{Definition}

\newtheorem{theorem}{Theorem}
\newtheorem{remark}{Remark}

\newcommand{\R}{\mathbb{R}}

\newcommand{\Splus}{\mathbb{S}^n_{>0}}
\newcommand{\Sskew}{\mathbb{S}^n_{\mathrm{skew}}}

\def\BibTeX{{\rm B\kern-.05em{\sc i\kern-.025em b}\kern-.08em
    T\kern-.1667em\lower.7ex\hbox{E}\kern-.125emX}}
\begin{document}
\title{\LARGE \bf
Distributions and Direct Parametrization for Stable Stochastic State-Space Models\\

}
\author{Mohamad Al Ahdab$^{1,2}$, Zheng-Hua Tan$^{1,2}$, John Leth$^{1}$ 
    \thanks{This project is supported by the Pioneer Centre for Artificial Intelligence, Denmark.}
	\thanks{$^{1}$ Department of Electronic Systems, Aalborg University, Aalborg Øst 9220, Denmark.
		{\tt\small \{maah,zh,jjl\}@es.aau.dk}} \thanks{$^{2}$Pioneer Centre for Artificial Intelligence, Copenhagen 1350, Denmark.}
}
\maketitle
\thispagestyle{empty}
\begin{abstract}
We present a direct parametrization for continuous-time stochastic state-space models that ensures external stability via the stochastic bounded-real lemma. Our formulation facilitates the construction of probabilistic priors that enforce almost-sure stability which are suitable for sampling-based Bayesian inference methods. We validate our work with a simulation example and demonstrate its ability to yield stable predictions with uncertainty quantification.
\end{abstract}

\pagestyle{empty}
\input{paper}
\thispagestyle{empty}
\begingroup
\bibliographystyle{IEEEtran}
\bibliography{literature}							
\endgroup

\end{document}

%% file: paper.tex
\section{Notation}
\label{sec:notation}
We let {$\mathbb{R}_{>0}=(0,\infty)$, $\mathbb{R}_{\geq0}=[0,\infty)$}, $\Splus$ denote the convex cone of positive definite matrices, and $\Sskew$ the vector space of $n\times n$ skew-symmetric matrices. 
Let \(\operatorname{vec}:\mathbb{R}^{m\times n}\to\mathbb{R}^{mn}\) be defined by mapping any \(X\in \mathbb{R}^{m\times n}\) to its column-stacked vector (with inverse \(\operatorname{vec}^{-1}\)). Similarly, define \(\operatorname{veck}:\Sskew \to\mathbb{R}^{\frac{n(n-1)}{2}}\) by mapping any skew-symmetric \(S\) to the vector of its below-diagonal entries (with inverse \(\operatorname{veck}^{-1}\)). For a complex vector \(v\in\mathbb{C}^n\), we let \(v^\dagger\) denotes its Hermitian (conjugate) transpose. We let \(\Gamma\) denote the Gamma function, which generalizes the factorial to complex numbers (with \(\Gamma(n)=(n-1)!\) for positive integers \(n\)). Finally, we let \(\mathrm{diag}_{m\times n}(a_1,\dots,a_{\mathrm{min}\{m,n\}})\) denotes the \(m\times n\) rectangular diagonal matrix with \(a_i\) in the \((i,i)\)th entry for \(1\leq i\leq\min\{m,n\}\) and zeros elsewhere. 

\section{Introduction}
State-space models (SSMs) provide a framework for representing dynamical systems and have broad applications across control engineering, signal processing, econometrics, and other scientific disciplines. Over the years, various system identification methods have been developed to estimate SSMs from experimental data under different settings. These include \emph{subspace identification} approaches, such as \cite{van1994n4sid,campi2006iterative}; \emph{prediction-error} or \emph{maximum-likelihood} methods \cite{mckelvey2004data,simpkins2012system,sato2019riemannian}; and \emph{Bayesian} techniques that quantify uncertainties in model parameters \cite{chiuso2016regularization, ljung2020shift}. Additionally, \emph{sequential Monte Carlo} (SMC) methods \cite{schon2015sequential} have demonstrated promise for complex inference tasks for SSMs. Moreover, the role of \emph{uncertainty quantification} has been gaining attention in both classical and Bayesian identification settings \cite{GRES2022108581,chiuso2016regularization}.

In many real-world applications, one may have strong prior knowledge that the underlying physical or engineered system must satisfy certain stability criteria. When such information is ignored, the identification method may yield models that violate the known stability criteria, thus resulting in poor performance and unreliable long-term predictions. Several approaches have been proposed in the literature to ensure stability or robust performance in the identified models. Constrained optimization strategies have been introduced for deterministic discrete-time systems \cite{lacy2003subspace}, while specialized regularization techniques have been employed for stochastic discrete-time SSMs \cite{van2001identification}. {Stability-constrained subspace methods have been discussed in, e.g., \cite{boots2007constraint, obara2023stable}, while kernel-based regularization techniques guarantee stability for deterministic continuous- and discrete-time systems in \cite{CHEN2018109}. Direct parametrization approaches have been developed to ensure asymptotic stability in deterministic autonomous systems \cite{obara2023stable} and quadratic stability in input–output deterministic discrete-time linear parameter-varying systems \cite{kon2024unconstrained}.
Yet, methods ensuring stability for continuous-time stochastic SSMs with state and input-dependent (multiplicative) noise are scarce despite their prevalence in diverse applications \cite{Network,DONNET2013929} and the control literature \cite{gravell2020learning}.}

In this paper, we address this gap by proposing a \emph{direct parametrization} for linear continuous-time \emph{stochastic} SSMs that automatically enforces \emph{external stability} (i.e., bounded input-output behavior). Our approach exploits the \emph{stochastic Bounded-Real Lemma (BRL)} , ensuring a finite input-output $\mathcal{L}^2$ gain. We additionally utilize the direct parametrization to construct \emph{probabilistic priors} for these externally stable SSMs by assigning distributions over the parameters in a manner that guarantees the external stability property \emph{almost surely}. This enables seamless integration into a Bayesian inference workflow, where standard sampling algorithms can be used without the need for rejection sampling or ad-hoc stability penalties. Moreover, our construction allows one to encode prior knowledge not only about stability but also about decay rates, oscillatory modes, and noise characteristics.

Overall, our contributions are threefold: (i) We derive a direct parametrization for continuous‑time stochastic SSMs with state‑ and input‑dependent noise that inherently satisfies the stochastic bounded‑real lemma; (ii) We design probabilistic priors over the parametrized model matrices using distributions that enforce stability and discuss sampling techniques from these distributions; and (iii) We demonstrate, through a simulation example, that the proposed framework yields reliable and stable predictions even when extrapolating beyond the data used for inference.
\section{Setup}
The description and definitions in this section are a summary of \cite{stochastic_Hinfty}.
In this paper, we consider linear stochastic SSMs in the form 
\begin{subequations}
\label{eq:SSM}
\begin{align}
\label{eq:sde}
&dx(t) = \left(A x(t) +Bu(t) \right)dt + \left[F x(t)~Gu(t)\right]dw(t),\\
\label{eq:output}
&y(t) = C x(t) + D u(t),
\end{align}
\end{subequations}
where \(A \in \mathbb{R}^{n\times n}, F\in \mathbb{R}^{n\times n}, G \in \mathbb{R}^{n \times \ell}, B \in \mathbb{R}^{n \times \ell}, C\in \mathbb{R}^{q \times n}, D \in \mathbb{R}^{q \times \ell})\) and \(w=\left[w_1~w_2\right]^\top\) with \(w_1\) and \(w_2\) being real scalar Wiener processes on a probability space \((\Omega, \mathcal{F}, \mathbb{P})\) relative to an increasing family
\((\mathcal{F}_t)_{t \ge 0}\) of \(\sigma\)-algebras \(\mathcal{F}_t \subset \mathcal{F}\). { All stochastic integrals are understood in the Itô sense.}
For the Wiener processes, we assume that \(\mathbb{E}\bigl[(w(s) - w(t))(w(s) - w(t))^\top\bigr] = \begin{bmatrix}
    1 & \rho\\ \rho & 1
\end{bmatrix}(s-t), \)
with \(\rho\in[-1,1]\).
We let \(\mathcal{L}^2(\Omega,\mathbb{R}^k)\) denote the space of square-integrable
$\mathbb{R}^k$-valued functions on the probability space \((Ω, \mathcal{F}, \mathbb{P})\) i.e., $f\in \mathcal{L}^2(\Omega,\mathbb{R}^k)$ if $\|f\|_{\mathcal{L}^2}=\mathbb{E}[{f^\top f}]<\infty$. For \(T\in\mathbb{R}_{>0}\), we denote by \(\mathcal{L}^2_w\left([0~T];\mathcal{L}^2(\Omega,\mathbb{R}^k)\right)\) the space of nonanticipative stochastic processes \(z(\cdot)=\left(z(t)\right)_{t\in[0~T]}\) with respect to \(\left(\mathcal{F}_t\right)_{t\in [0,T]}\) satisfying
\[
\|z\|^2_{\mathcal{L}^2_w} = \int^T_{0} 
\|z(t)\|_{\mathcal{L}^2}dt
<\infty.
\]
For any \(u\in\mathcal{L}^2_w([0~T];\mathcal{L}^2(\Omega,\mathbb{R}^\ell))\), and \(x_0\in \mathbb{R}^n\), we denote \(x(t;x_0,u)\) and \(y(t;x_0,u)\) as the solution for \eqref{eq:SSM} with initial condition \(x(0)=x_0\).
The Stochastic Differential Equation (SDE) \eqref{eq:sde} with initial condition \(x(0)=x_0\) and input \(u\in\mathcal{L}_{w}([0~T];\mathcal{L}^2(\Omega,\mathbb{R}^\ell))\) has a unique solution \(x(\cdot)=x(\cdot;x_0,u)\in \mathcal{L}_{w}([0,T],\mathcal{L}^2(\Omega,\mathbb{R}^n))\) \cite{stochastic_Hinfty}. We now define \emph{internal stability} for the SSM in \eqref{eq:SSM}. 
\begin{definition}[Internal Stability]
    The SSM in \eqref{eq:SSM} is said to be \emph{internally stable} if there exist constants \(c,\alpha>0\) such that the unforced solution of the SDE \eqref{eq:sde} \(x(\cdot)=x(\cdot; x_0, 0)\) (\(u=0\)) with \(x(0)=x_0\in\mathbb{R}^n\) satisfies
    \begin{equation}
    \label{eq:decay_rate}
    \mathbb{E}\left[\|x(t)\|^2\right]\leq c\exp(-\alpha t)\|x_0\|^2,~\text{for all } t>0.
    \end{equation}
    In other words, the SDE \eqref{eq:sde} is mean-squared exponentially stable. In addition, we call the pair \((A,F)\) a stable pair.
\end{definition}
\begin{proposition}[\cite{stochastic_Hinfty}]
\label{prob:Stable_AF}
    For \(Q\in \Splus\), the SSM in \eqref{eq:SSM} is internally stable if and only if there exists \(P\in\Splus\) such that
    \begin{equation}
    \label{eq:Lyap_AF}
    A^{\top}P+PA + F^\top P F = - Q
    \end{equation}
\end{proposition}
Note that if \(P\) exists then it is unique, and when \(Q=\alpha P\), the unforced solution of the SDE \eqref{eq:sde} \(x(\cdot)=x(\cdot; x_0, 0)\) will satisfy 
\(\mathbb{E}\left[\|x(t)\|^2\right]\leq \frac{\lambda_{\mathrm{max}}(P)}{\lambda_{\mathrm{min}}(P)}\exp(-\alpha t)\|x_0\|^2,~\text{for all } t>0.\)

\begin{definition}[External Stability]
The SSM \eqref{eq:SSM} is said to be \emph{externally stable} or $\mathcal{L}^2$ \emph{input-output stable} if, for every \(u \in \mathcal{L}_w^2({\mathbb{R}_{\geq 0}}; \mathcal{L}^2(\Omega,\mathbb{R}^\ell))\), we have \(y(\cdot)=y(\cdot;0,u) \in \mathcal{L}_w^2({\mathbb{R}_{\geq 0}}; \mathcal{L}^2(\Omega,\mathbb{R}^q))\)
and there exists a constant $\gamma \ge 0$ such that 
\begin{equation}
\label{eq:BIBO}
\|y(t)\|_{\mathcal{L}_w^2} \leq \gamma \,\|u(t)\|_{\mathcal{L}_w^2}.
\end{equation}
\end{definition}
If the SSM \eqref{eq:SSM} is externally stable, then we can define an input-output operator \(\mathbb{L}: \mathcal{L}_w^2({\mathbb{R}_{\geq 0}}; \mathcal{L}^2(\Omega,\mathbb{R}^\ell)) \!\to\! \mathcal{L}_w^2({\mathbb{R}_{\geq 0}};\mathcal{L}^2(\Omega,\mathbb{R}^q)),\) via \(u\!\mapsto \!y(\cdot;0,u)\) (see \cite{stochastic_Hinfty}). The induced norm on this operator \(\|\mathbb{L}\|:= \sup_{\|u\|_{\mathcal{L}^2_w}=1}\|y(\cdot;0,u)\|_{\mathcal{L}^2_{w}}\) represents the minimal \(\gamma>0\) such that \eqref{eq:BIBO} is satisfied. This norm is analogous to the \(\mathcal{H}_{\infty}\)-norm. Therefore, it is often referred to as the stochastic \(\mathcal{H}_{\infty}\)-norm. In \cite{stochastic_Hinfty}, it is shown that if \eqref{eq:SSM} is internally stable, then it is also externally stable and there exists a constant \(\gamma>0\) such that \(\|\mathbb{L}\|<\gamma\).

\begin{theorem}[Stochastic Bounded Real Lemma \cite{stochastic_Hinfty}]
\label{th:Stochastic_Bounded_Real_Lemma}
For \(\gamma>0\), the following statements are equivalent 
\begin{enumerate}
    \item The SSM in \eqref{eq:SSM} is internally stable and \(\|\mathbb{L}\|< \gamma\).
    \item There exists \(P\in \Splus\) such that 
    \begin{equation}
    \label{eq:LMI_BRL}
        M:=\begin{bmatrix}
M_{11}
& M_{12} \\
M^\top_{12} 
& M_{22} 
\end{bmatrix}<0
    \end{equation}
\end{enumerate}
with \(\mathrm{M}_{11}= PA + A^\top P + F^\top PF + C^\top C, M_{12}=PB + \rho F^\top PG + C^\top D,\) and \(M_{22}= D^\top D - \gamma^2 I + G^\top PG\).
\end{theorem}
Note that for \eqref{eq:LMI_BRL} to be true, the SSM \eqref{eq:SSM} must be internally stable \cite{stochastic_Hinfty}. Additionally, for an SSM to satisfy the LMI \eqref{eq:LMI_BRL}, it is assumed to be stabilizable and detectable. 
\section{Parametrizations and Distributions}
In this section, we will describe probability distributions for the matrices of the SSM in \eqref{eq:SSM} defined on a distinct probability space \((\Omega', \mathcal{F}', \mathbb{P}')\) separate from the original space \((\Omega, \mathcal{F}, \mathbb{P})\) associated with the Wiener processes \(w_1, w_2\). Specifically, we first construct a probability distribution over stable \((A,F)\) pairs by requiring that the Lyapunov inequality \eqref{eq:Lyap_AF} holds almost surely under \(\mathbb{P}'\). Next, we define a probability distribution for the full SSM matrices \((A,B,C,D,F,G)\) and \(\gamma\) such that the LMI condition \eqref{eq:LMI_BRL} in Theorem \ref{th:Stochastic_Bounded_Real_Lemma} is satisfied almost surely.

\subsection{Distributions over stable \((A,F)\) pairs}
\label{sec:AF-distributions}
We first state the following proposition which provides a direct parametrization for all stable \((A,F)\) pairs.
\begin{proposition}
    \label{prob:Stable_AF_param}
    For every stable pair \((A,F)\), there exist
    \(P,Q\in \Splus\), \(S\in \Sskew\), and \(\tilde{F} \in \mathbb{R}^{n\times n}\) such that
    \begin{equation}
    \label{eq:Stable_AF_param}
    A= -\frac{1}{2}P^{-1}(Q+\tilde{F}^\top\tilde{F}+S),~F = L_{P^{-1}}\tilde{F},
    \end{equation}
    where \(L_{P^{-1}}\) is the lower-triangular Cholesky factor of \(P^{-1}\). Conversely, if \(A\) and \(F\) are defined by \eqref{eq:Stable_AF_param} for some \(P,Q\in \Splus\), \(S\in \Sskew\), and \(\tilde{F}\in \R^{n\times n}\), then the pair \((A,F)\) are stable.
\end{proposition}
{\begin{proof}
Let \((A,F)\) be a stable pair, then according to Proposition \ref{prob:Stable_AF}, for any given \(Q\in \Splus\), there exists a \(P\in \Splus\) such that 
\(-Q = A^{\top}P+PA + F^\top P F 
=\tilde{A}^\top + \tilde{A} + \tilde{F}^\top \tilde{F}\)
where the second equality follows by defining \(\tilde{A}:=PA\) and \(\tilde{F}:=L_{P^{-1}}^{-1}F\). With \(S\in \Sskew\) representing the free skew-symmetric part of \(\tilde{A}\) it follows that \(\tilde{A}=-\frac{1}{2}\left(Q+\tilde{F}^\top \tilde{F} + S\right)\) and therefore \(A=-\frac{1}{2}P^{-1}\left(Q+\tilde{F}^\top \tilde{F} + S\right)\). The converse can be shown by a direct substitution of \eqref{eq:Stable_AF_param} into \eqref{eq:Lyap_AF}.
\end{proof}}
It is important to remark that according to Proposition \ref{prob:Stable_AF}, the matrix \(Q\) can be fixed to any value. If we fix \(Q=\alpha P\), then we can obtain a parametrization which also considers the decay rate \(\alpha \) in \eqref{eq:decay_rate}. Additionally, {note} that for any eigenvalue \(\lambda\in \mathbb{C}\) of \(A\) with eigenvector \(v\in \mathbb{C}^n\), we have { \(-\frac{1}{2}P^{-1}(Q+ \tilde{F}^\top \tilde{F}+ S)v =\lambda v \Rightarrow \frac{-1}{2}(Q+\tilde{F}^\top \tilde{F}+S)v= \lambda Pv\), which means}
\begin{equation}
    \lambda = \frac{-1}{2}\left(\frac{v^\dagger (Q+\tilde{F}^\top \tilde{F})v}{v^\dagger P v}+\frac{v^\dagger Sv}{v^\dagger P v} \right),
\end{equation}
with \(\mathrm{Re}(\lambda)\!=\! \frac{-v^\dagger (Q+\tilde{F}^\top \tilde{F})v}{2v^\dagger P v}\), and \(\mathrm{Im}(\lambda)\!=\! \frac{-v^\dagger Sv}{2v^\dagger P v}\).
Based on the parametrization in \eqref{eq:Stable_AF_param}, we can define a distribution on stable \((A,F)\) pairs by defining a joint distribution over the parameters \(P,\tilde{F},S\), and either fixing \(Q\in \Splus\), include it as a random matrix with the joint distribution with the other parameters, or letting \(Q=\alpha P\) and including \(\alpha\) in the joint distribution for the other parameters. In this paper, we will focus on three possible distributions for \((A,F)\).
\begin{definition}
    \label{def:WSN-distributions}
    Let \(P^{-1}\sim \mathrm{Wishart}(k_p,\Sigma_p)\) with \(k_p\geq n\), \(\mathrm{vec}(\tilde{F})\sim \mathcal{N}(\mu_f,\Sigma_f) \), and \(\mathrm{veck}(S)\sim \mathcal{N}(\mu_s,\Sigma_s)\) with \(P^{-1},\tilde{F}\) and \(S\) being independent. Then we say the pair \((A,F)\) follows
    \begin{itemize}
        \item The Wishart-Normal-Skew (WNS) distribution if \(Q\in\Splus \) is fixed and \((A,F)\) are obtained as in \eqref{eq:Stable_AF_param}.
        \item The \(Q\)-WNS distribution if \(Q\sim \mathrm{Wishart}(k_q,\Sigma_q) \) independently from \(P^{-1},S,\tilde{F}\), and \((A,F)\) are obtained as in \eqref{eq:Stable_AF_param}.
        \item The \(\alpha P\)-WNS if \(\alpha\sim \mathrm{Gamma}(k_\alpha,\theta_\alpha )\) independently from \(P^{-1},S,\tilde{F}\), \(Q=\alpha P\), and \((A,F)\) are obtained as in \eqref{eq:Stable_AF_param}.
    \end{itemize}
\end{definition}
\begin{remark}
    One reason to use the Wishart distribution for \(P^{-1}\sim \mathrm{Wishart}({k_p,\Sigma_p})
    \) is the fact that it can be sampled using the Bartlett decomposition \(P^{-1}=L_{\Sigma_p}\tilde{L}_{P^{-1}}\tilde{L}^\top_{P^{-1}}L^\top_{\Sigma_p}\), where \(L_{\Sigma_p}\) is the Cholesky factor of \(\Sigma_p\) and 
\begin{equation}
\label{eq:Bartlett}
\tilde{L}_{P^{-1}} = \begin{bmatrix}
c_{1} & 0            & \cdots & 0 \\
m_{21} & c_{2}       & \cdots & 0 \\
\vdots & \vdots  & \ddots & 0 \\
m_{n1} & m_{n2} & \cdots & c_{n},\end{bmatrix}
\end{equation}
where \(c_i^2 \sim \chi^2_{n-i+1}\) and \(m_{ij} \sim N(0,1)\) independently \cite{Bartlett}. Therefore, we can sample the Cholesky factor \(L_{P^{-1}}=L_{\Sigma_p}\tilde{L}_{P^{-1}}\) of \(P^{-1}\) without having to sample \(P^{-1}\) and then compute a Cholesky factor for it. This approach avoids the computational cost of calculating a Cholesky factor.
\end{remark}
The expectations of \(A\) and \(F\) are  
\begin{subequations}
    \label{eq:Expectations}
    \begin{align}
    \label{eq:Expt_A}
    &\mathbb{E}\left[A\right] = -\frac{1}{2}\left(\mathbb{E}\left[P^{-1}\right]\right)\left(\mathbb{E}\left[Q\right]+\mathbb{E}\left[\tilde{F}^\top\tilde{F}\right]+\mathbb{E}\left[S\right]\right),\\
    &\mathbb{E}\left[F\right] = \mathbb{E}\left [L_{P^{-1}}\right] \mathbb{E}\left[ \tilde{F}\right],
    \end{align}
\end{subequations}
where \(\mathbb{E}\left[P^{-1}\right]\!= \!k_p \Sigma_p, 
~\mathbb{E}\left[S\right]\!=\! \mathrm{veck}^{-1}(\mu_s),~\mathbb{E}[\tilde{F}]\!=\! \mathrm{vec}^{-1}(\mu_f), ~\mathbb{E}\left[L_{P^{-1}}\right]\!=\! L_{\Sigma_p}\mathsf{D}\) {(using \eqref{eq:Bartlett})}
with \(\Sigma_p \!=\!L_{\Sigma_p}L_{\Sigma_p}^\top\)
and \(\mathsf{D}\) a diagonal matrix with entries
\[\textstyle
\mathsf{D}_{ii}={\mathbb{E}[c_i]}=\sqrt{2}\frac{\Gamma\left(\frac{{n}-i+2}{2}\right)}{\Gamma\left(\frac{{n}-i+1}{2}\right)}. 
\]

Furthermore, we have
\(
\mathbb{E}[\tilde{F}^\top \tilde{F}]\!=\! \mathbb{E}[\tilde{F}]^\top\mathbb{E}[\tilde{F}]\!+\! \mathbf{\Sigma}^F,
\)
where \(\mathbf{\Sigma}^F_{ij}\!=\! \mathrm{tr}\Bigl(\Sigma^{ij}_f\Bigr)
\)
with \(\Sigma^{ij}_f\) denoting the \((ij)\)-th \(n\!\times \!n\) block of \(\Sigma_f\).

Finally, the expectation of \(Q\) is given by
\[
\mathbb{E}[Q] =
\begin{cases}
Q, & \text{(WSN)},\\
k_q\Sigma_q, & \text{(}Q\text{-WSN)},\\
 \frac{k_\alpha \theta_{\alpha}}{k_p-n-1} \Sigma^{-1}_{P}, & \text{(}\alpha P\text{-WSN}, k_p>n+1\text{)}.
\end{cases}
\]
The expectation formulas in \eqref{eq:Expectations} and the expression for the eigenvalues of \(A\) reveal useful insights into how one may choose the parameters for the proposed distributions. Recall that the real part of any eigenvalue \(\lambda\) of \(A\) in \eqref{eq:Stable_AF_param} depends on \(Q + \tilde{F}^\top\tilde{F}\), while the imaginary part is influenced by the skew-symmetric matrix \(S\). 
For instance, when adopting the WNS distribution, one fixes \(Q \in \Splus\) and draws samples for \(P^{-1}\), \(\tilde{F}\), and \(S\). This approach ties the mean-squared decay properties (through \(Q\)) to a known baseline while permitting flexibility in both the real part of the eigenvalues (through the distribution of \(P\) and \(\tilde{F}\)) and their oscillatory behavior (through \(S\)). Suppose one wants to quantify the uncertainty in the mean-squared decay properties with the WNS. In that case, it is possible by either modifying the uncertainty in \(P^{-1}\) through \(\Sigma_p\) and \(k_p\) or modifying the uncertainty in \(\tilde{F}\) through \(\Sigma_f\). However, this will also alter the uncertainty in the drift part's oscillatory behavior or the diffusion part's uncertainty. In contrast, the \(Q\)-WNS distribution treats \(Q\) itself as random, thereby enabling us to quantify the uncertainty in mean-square decay properties without affecting the uncertainty in the oscillatory properties or the diffusion term. 

Similarly, the \(\alpha P\)-WNS construction, where \(Q = \alpha P\) and \(\alpha \sim \mathrm{Gamma}(k_\alpha,\theta_\alpha)\) independently, allows us to directly randomize the decay rate \(\alpha\) without affecting the oscillatory properties and the diffusion term. Nevertheless, both \(Q\)-WNS and \(\alpha P\)-WNS become an over parametrization to the pair \((A,F)\) according to Proposition \ref{prob:Stable_AF}. 


\subsection{Distributions for SSMs satisfying the stochastic BRL}
We start by stating a parametrization for an SSM on the form \eqref{eq:SSM} to satisfy the stochastic BRL.
\begin{proposition}
\label{prob:BRL_param}
The matrices \(A,B,C,D,F,G\) of any SSM on the form \eqref{eq:SSM} which satisfies the stochastic BRL with \(\|L\|<\gamma\) can be written as 
\begin{subequations}
\label{eq:BRL_param}
\begin{align}
        \label{eq:AF_external_stable}
        & A \!=\! -\frac{1}{2}P^{-1}(Q+\tilde{F}^\top\tilde{F}+C^{\top} C + S), ~F\!=\!L_{P^{-1}} \tilde{F},\\
        \label{eq:G_external_stable}
        & G\!=\!L_{P^{-1}}\tilde{G},~B\!=\! P^{-1}\left(L_Q\tilde{B}-\rho \tilde{F}^\top \tilde{G} - C^\top D\right),
\end{align}
\end{subequations}
with \(P,Q\in \Splus\), \(S\in \Sskew\), \(\tilde{F} \in \mathbb{R}^{n\times n}\), \(C \in \mathbb{R}^{q\times n}\), \(D\in \mathbb{R}^{q\times \ell },~ \tilde{B},\tilde{G} \in \mathbb{R}^{n\times \ell}\) satisfying
\begin{equation}
    \label{eq:param_cond_BRL}
    \tilde{B}^\top \tilde{B} + D^\top D + \tilde{G}^\top \tilde{G} < \gamma^2 I,
    \end{equation}
and \(L_{P^{-1}},L_{Q}\) are the lower-triangular Cholesky factors of \(P^{-1}\) and \(Q\), respectively.
Conversely, any choice of \(P,Q\in \Splus\), \(S\in \Sskew\), \(\tilde{F} \in \mathbb{R}^{n\times n}\), \(C \in \mathbb{R}^{q\times n}\), and \(D\in \mathbb{R}^{q\times \ell },~ \tilde{B},\tilde{G} \in \mathbb{R}^{n\times \ell}\) satisfying condition \eqref{eq:param_cond_BRL} with \(A,F,G,B\) defined as in \eqref{eq:BRL_param} will yield an SSM satisfying the stochastic BRL with \(\|\mathbb{L}\|<\gamma\).
\end{proposition}
\begin{proof}
Using Theorem \ref{th:Stochastic_Bounded_Real_Lemma} and the Schur complement, we know that 
\begin{equation}
\label{eq:Shur_LMI}
    M<0 \Leftrightarrow M_{11}< 0, ~  M_{22}-M^\top_{12}M^{-1}_{11}M_{12}<0.
\end{equation}
{ Now define \(Q\in \Splus\) by \(-Q = M_{11}=PA+A^\top P + F^\top P F + C^\top C\). Then \(PA+A^\top P + F^\top P F =-(Q + C^\top C)\) is a Lyapunov equation for internal stability \eqref{eq:Lyap_AF}.} Therefore, using Proposition \ref{prob:Stable_AF_param}, we can write \(A=-\frac{1}{2}P^{-1}(Q+C^\top C + \tilde{F}^\top \tilde{F}+S)\) with \(\tilde{F}\in \mathbb{R}^{n\times n}\), \(S\in \Sskew\), and \(F=L_{P^{-1}}\tilde{F}\).


We now show that for any SSM with matrices satisfying \eqref{eq:LMI_BRL}, we can find \(\tilde{B},\tilde{G}\) with \(\tilde{B}^\top \tilde{B}+D^\top D+\tilde{G}^\top\tilde{G}<\gamma^2 I\). Using \eqref{eq:Shur_LMI}, we know that \(M_{22}-M^\top_{12}M_{11}^{-1}M_{12}<0 \) and \(M_{11}\!=\!-Q\!<\!0\). Now let \(\tilde{B} = L^{-1}_{Q}\left(PB + \rho \tilde{F}^\top \tilde{G} + C^\top D\right)\) and {\(\tilde{G}=L^{-1}_{P^{-1}}G\)}, then \eqref{eq:G_external_stable} follows, and by substituting \(\tilde{F},\tilde{G},\tilde{B}\) in {\(M_{22}-M_{12}^\top M^{-1}_{11}M_{12}\)}, we get the condition \eqref{eq:param_cond_BRL}. 
For the converse, direct substitution of \(A,F,G,B\) defined in \eqref{eq:BRL_param} with any \(P,Q\in \Splus\), \(S\in \Sskew\), \(\tilde{F} \in \mathbb{R}^{n\times n}\), \(C \in \mathbb{R}^{q\times n}\), and \(D\in \mathbb{R}^{q\times \ell },~ \tilde{B},\tilde{G} \in \mathbb{R}^{n\times \ell}\) satisfying the condition \eqref{eq:param_cond_BRL}, will produce a matrix \(M\) as in  \eqref{eq:LMI_BRL} satisfying \(M < 0\).
\end{proof}
{Note that the $Q$ in Proposition \ref{prob:BRL_param} is a positive‑definite slack matrix that is different from the $Q$ in Proposition~\ref{prob:Stable_AF_param}.}

Using the parametrization in Proposition \ref{prob:BRL_param}, we now construct probability distributions for the matrices of the SSM in \eqref{eq:SSM} satisfying the stochastic BRL with \(\|\mathbb{L}\|< \gamma\). To do so, we assign a joint distribution to the free parameters \(P\), \(Q\), \(S\), \(\tilde{F}\), \(\tilde{G}\), \(C\), and \(D\) (as they appear in \eqref{eq:BRL_param} and \eqref{eq:LMI_BRL}), while ensuring that the constraint in \eqref{eq:param_cond_BRL} is met. For example, one may adopt the distributions for the pair \((A,F)\) presented in Section \ref{sec:AF-distributions} and extend them by including an independent distribution for \(C\). Since \(C\) is free in \eqref{eq:BRL_param}, one may specify
\(
\mathrm{vec}(C)\sim \mathcal{N}(\mu_c,\Sigma_c),
\)
analogous to the distribution chosen for \(\tilde{F}\).

For the parameters \(\tilde{B}\), \(D\), and \(\tilde{G}\), observe that the condition \eqref{eq:param_cond_BRL} is equivalent to
\begin{equation}
\label{eq:Z-definition}
Z^\top Z < I,\quad \text{with} \quad Z:=\frac{1}{\gamma}\begin{bmatrix}\tilde{B}^\top~D^\top ~\tilde{G}^\top \end{bmatrix}^\top.
\end{equation}

Thus, one may generate \(\tilde{B}\), \(D\), and \(\tilde{G}\) by first sampling a matrix \(Z\) from a distribution supported on the unit ball with respect to the spectral-norm. However, sampling \(Z\) arbitrarily on this ball may yield values such that \(Z^\top Z\) is considerably smaller than the identity (i.e., \(Z^\top Z<\gamma_\epsilon I\) with \(\gamma_\epsilon\ll1\)), which in turn renders \(\gamma\) a conservative upper bound for \(\|\mathbb{L}\|\).

To address this issue, we instead aim to sample \(Z\) so that its maximum singular value \(\sigma^z_1\) is close to 1. Specifically, we let \(\sigma^z_1=(1-\varepsilon)\)
with \(0<\varepsilon<<1\). To implement this, we consider the Singular Value Decomposition (SVD) of \(Z\):
\[
Z = U_z\,\mathrm{diag}_{2n+q,\ell}(\sigma^z_1,\dots,\sigma^z_k)V_z^\top,~k=\min(2n+q,\ell),
\]
with the maximum singular value \(\sigma^z_1\) fixed to \((1-\varepsilon)\) and the rest are independently sampled from the uniform distribution
\(
\sigma^z_i\sim \mathrm{U}\bigl(0,\,1-\varepsilon\bigr), 
\) for \(i\in\{2\dots,k\}\).
For the orthonormal matrices \(U_z\) and \(V_z\), one practical approach is to sample
\(
\mathrm{vec}(\tilde{U}_z)\sim \mathcal{N}(0,I\otimes \Sigma_{U_z})\text{ and } \mathrm{vec}(\tilde{V}_z)\sim \mathcal{N}(0,I\otimes \Sigma_{V_z}),
\)
and then compute \(U_z\) and \(V_z\) as the orthogonal factors in the QR-factorizations of \(\tilde{U}_z\) and \(\tilde{V}_z\), respectively. The resulting distribution over \(U_z\) and \(V_z\) is known as the matrix Angular Central Gaussian (ACG) distribution \cite{chikuse1990matrix}.
Another approach to sample \(U_z\) and \(V_z\) is to sample them using Cayley transform \cite{gallier2006remarks} \(U_z\!=\!E_u(I-S_u)(I+S_u)^{-1}\) and \(V_z\!=\!E_z(I-S_v)(I+S_v)^{-1}\), where \(\mathrm{veck}(S_u)\sim \mathcal{N}(\mu_u,\Sigma_u)\), \(\mathrm{veck}(S_v)\sim \mathcal{N}(\mu_v,\Sigma_v)\), and \(E_u,E_v\) are diagonal matrices with entries \(\pm1\) sampled uniformly.

In the special case of a single input (\(\ell=1\)), the matrix \(Z\) reduces to a column vector, so its SVD takes the form
\(
Z=\sigma_1^z u_z,
\)
with \(u_z\) a unit vector obtained by normalizing a Gaussian vector \(\tilde{u}_z\sim \mathcal{N}(0,\Sigma_{u_z})\). The distribution of \(u_z\) then follows the (vector) ACG distribution. 

If we want to have more flexibility in specifying the uncertainty in the individual norms of \(\tilde{B}\), \(\tilde{G}\), and \(D\), we sample them based on their SVDs (similar to \(Z\)) with the maximum singular values set to \(\lambda_b(1-\varepsilon)\gamma\), \(\lambda_g(1-\varepsilon)\gamma\), and \(\lambda_d(1-\varepsilon)\gamma\), respectively. Here, the weights \(\lambda_b\), \(\lambda_g\), and \(\lambda_d\) are strictly positive and satisfy
\(\lambda_b + \lambda_g + \lambda_d = 1\). These weights can be drawn from a Dirichlet distribution
\(
\lambda = \left[\lambda_b,\lambda_g,\lambda_d\right] \sim \mathrm{Dir}(\alpha_1, \alpha_2, \alpha_3),
\)
with \(\alpha_1,\alpha_2,\alpha_3 > 0\). 

We summarize the extension of the distributions defined in \eqref{def:WSN-distributions} to the case of internal stability with \(\|L\|<\gamma\) as follows

\begin{definition}[WNS-BRL Distributions]
\label{def:WSN-BRL-distributions}
Let \(\gamma\!>\!0\), \(P^{-1}\!\sim\! \mathrm{Wishart}(k_p,\Sigma_p)\) with \(k_p\!\geq\!n\), \(\mathrm{vec}(\tilde{F})\!\sim\! \mathcal{N}(\mu_f,\Sigma_f)\); \(\mathrm{vec}(C)\sim \mathcal{N}(\mu_c,\Sigma_c)\), and \(\mathrm{veck}(S)\!\sim\! \mathcal{N}(\mu_s,\Sigma_s)\). Let \(
Z=U_z\,\mathrm{diag}_{2n+q,\ell}(\sigma^z_1,\dots,\sigma^z_k)V^\top_z,\) with \(\sigma^z_1=1-\varepsilon\) for a fixed constant \(0<\varepsilon\ll1\), and the remaining singular values \(\sigma^z_2,\dots,\sigma_k^z\) are drawn independently from \(\mathrm{U}(0,1-\varepsilon)\) with \(k=\min(2n+q,\ell)\). The orthonormal matrices \(U_z\) and \(V_z\) are sampled from a matrix ACG distribution (or via a Cayley transform based on Gaussian skew-symmetric matrices and diagonal matrices with \(\pm1\) entries sampled uniformly), and the matrices \(\tilde{B}\), \(D\), and \(\tilde{G}\) are specified as in \eqref{eq:Z-definition}. If all the random variable are independent with \(A,B,C,D,F,G\) defined via \eqref{eq:BRL_param}, then resulting SSM in \eqref{eq:SSM} is said to follow
\begin{itemize}
    \item The WNS-BRL-ACG (WNS-BRL-Cayley) distribution if \(Q\in\Splus\) is fixed and \(U_z,V_z\) follow the ACG (or Cayley) distribution.
    \item The \(Q\)-WNS-BRL-ACG (\(Q\)-WNS-BRL-Cayley) distribution if \(Q\sim \mathrm{Wishart}(k_q,\Sigma_q)\) independently of the other variables.
    \item The \(\alpha P\)-WNS-BRL-ACG (\(\alpha P\)-WNS-BRL-Cayley) distribution if \(\alpha\sim \mathrm{Gamma}(k_\alpha,\theta_\alpha)\) independently with \(Q=\alpha P\).
\end{itemize}
Furthermore, if we choose to sample \(\tilde{B},\tilde{G}\) and \(D\) using their corresponding SVD with weights sampled from a Dirichlet distribution, then we call the distributions by the Dirichlet variants of their names.
\end{definition}
For these distributions, the expected value of \(A\) is the same as in \eqref{eq:Expt_A} but with the inclusion of \(\mathbb{E}\left[C^\top C\right]\). As for \(D,B\) and \(G\), the expectation is more involved and cannot be found in closed form for most of the cases. Nevertheless, the Dirichlet variants of the distributions provide us with the flexibility of specifying our prior belief on the size of the norms for the input-related matrices. 
{
\begin{remark}  
Even with no prior knowledge of the SSM, Proposition 3 provides an unconstrained parametrization for maximum likelihood estimation without the need for distributions. Similar to \cite{kon2024unconstrained}, positive scalars can be parametrized with an exponential function (e.g., \(\alpha =\mathrm{e}^{\bar{\alpha}}, \bar{\alpha}\in\mathbb{R}\)) and bounded scalars by a sigmoid function (e.g., \(\sigma_z^i=\frac{1-\varepsilon}{1+\mathrm{e}^{-{\bar{\sigma}^i_z}}},\bar{\sigma}_z^i\in\mathbb{R}\)); positive-definite matrices via Cholesky factors postive diagonal entries; and orthogonal matrices through a Cayley transform. Alternatively, one may instead perform Riemannian gradient descent on the manifolds of positive-definite, skew-symmetric and orthogonal matrices, similar to \cite{obara2023stable}.
\end{remark}}

\section{Bayesian Inference Example}
\begin{figure*}[t]  
    \centering

    \subfloat[\label{fig:avg_plot}]%
    {\includegraphics[width=.43\textwidth]{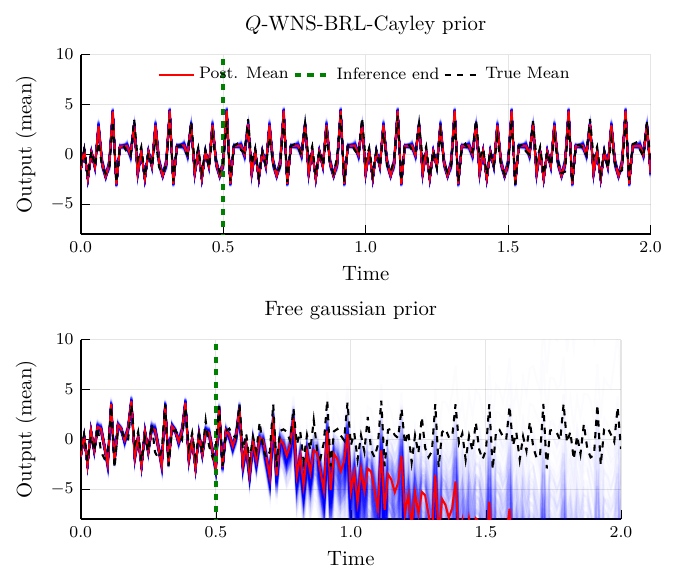}}
    ~
    \subfloat[\label{fig:var_plot}]%
    {\includegraphics[width=.43\textwidth]{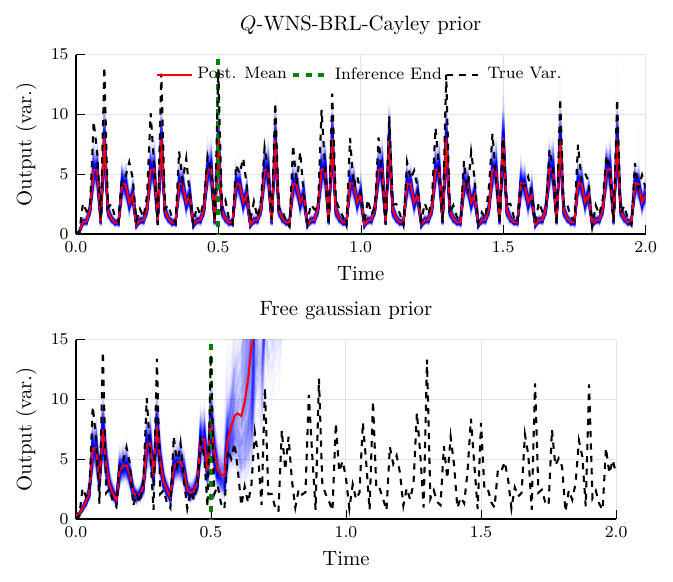}}

    \caption{The posterior mean (in red) and the posterior samples (in blue) of the mean (variance) over 100 realizations of the SSM in \eqref{eq:SSM} compared to the mean (variance) over 100 realizations from the true SSM (in black).}
    \label{fig:mean_and_var}
\end{figure*}
In this section, we illustrate how the prior distributions from Definition \ref{def:WSN-BRL-distributions} can be employed in Bayesian inference. In our setting, we assume that the true SSM~\eqref{eq:SSM} is known to satisfy the stochastic BRL with \(\|\mathbb{L}\| < \gamma\), where \(\gamma>0\) may itself be modeled as a random variable to capture uncertainty in the system gain.

Let \(\{t_i\}_{i=1}^N \subset [0,T]\) denote the measurement time instances. At each \(t_i\), we record the input \(u(t_i)\) and the corresponding noisy output \(\tilde{y}(t_i)\), where
\(
\tilde{y}(t_i)= y(t_i;x_0,u) + \xi(t_i).
\)
Here, \(y(t;x_0,u)\) is the output of the SSM~\eqref{eq:SSM} corresponding to the initial condition \(x(0)=x_0\), and \(\xi(t_i)\) represents measurement noise. We assume that the noise \(\xi(t_i)\) is independent and identically distributed \(\xi(t_i)\sim \mathcal{N}(0,\sigma^2I_{q})\).
Furthermore, we assume that the data is collected from \(M\) independent realizations (sample paths) of the SSM, all starting from the same initial condition \(x_0\). With this setup, Bayesian inference can be performed to update the prior distributions over the SSM parameters based on the observed data \(\{(u(t_i),\tilde{y}(t_i))\}_{i=1}^N\) from the \(M\) experiments. 


To generate data, we sampled a model with \({n=4},q=1,\ell=2\) using the \(Q\)-WNS-BRL-Cayley distribution with \(\Sigma_p\!=\!I_n, \Sigma_q\!=\! 2I, k_p\!=\!6, k_q\!=\!6, \mu_s\!=\!0, \Sigma_s\!=\!0.01I, \mu_f\!=\!0, \Sigma_f\!=\!2I, \mu_c\!= \!0, \Sigma_c\!=\!2I,\varepsilon =10^{-4}, \Sigma_u\!=\!I\), and \(\Sigma_{v}\!=\!I\).
We also set \(\rho\!=\!0.3\), \(\gamma\!=\!3\), and the measurement error variance \(\sigma\!=\!0.15\). We took {\(M\!=\!10\)} experiments and generated the inputs using 6 Fourier terms with random coefficients. The priors used for the inference were chosen to be different from the ones used for sampling the model to generate the data. Specifically, we chose the \(Q\)-WNS-BRL-Cayley with the other different choices being \(k_p=4, k_q = 4,\Sigma_q=I\), and \(\Sigma_s=I\). In addition, We assigned a prior \(\rho \sim \mathrm{U}(-0.95,0.95)\) and \(\gamma \sim \mathrm{Gamma}(1,1)\) (\(\mathbb{E}[\gamma]=1\)). 

To perform the inference, we used a Markov Chain Monte Carlo (MCMC) \cite{geyer1992practical} method {to obtain samples from the posterior distribution given the data}. The MCMC method we used is Hamiltonian Monte Carlo (HMC) with a step size of 0.01 and a number of steps being 3 \cite{neal2011mcmc}. We implemented the inference using the Turing.jl package in Julia \cite{ge2018t}. The code used for generating the results is available at Github \url{https://github.com/MOHAMMADZAHD93/StableSSMs}.

To compare, we used HMC with the same settings but with priors assigned directly on the SSM matrices without any parametrization. This approach can cause numerical issues {during the sampling process of an MCMC strategy since proposed SSM samples can lead to unstable trajectories}. To deal with this, every time a proposed SSM sample results in an unstable trajectory, we assign a log-likelihood of \(-\infty\) to that particular sample of matrices. That is, we force a rejection of the proposed sample in the MCMC step. Fig. \ref{fig:mean_and_var} demonstrate the results. 
In Fig. \ref{fig:avg_plot}, we compare the posterior mean and 500 posterior trajectory samples of the mean output (computed over 100 sample paths) for the two prior modeling strategies. In each case, we generated the results by{ simulating 100 trajectories for each one of the last 1000 posterior SSM samples obtained by the HMC and computing their sample mean trajectory}. {The last 500 mean trajectories obtained from the last 500 SSM posterior samples out of the last 1000 samples are overlaid} (in blue) along with the true mean trajectory (in black) which was generated using the actual model parameters. 

As seen in the top panel of Fig.\ref{fig:avg_plot}, the trajectories obtained under the \(Q\)-WNS-BRL-Cayley prior remain close to the true mean throughout the entire inference and prediction horizon. Even after the inference period ends (delineated by the green dashed line), the simulated trajectories exhibit bounded and stable behavior. This observation highlights that incorporating stability-inducing priors, done here through the \(Q\)-WNS-BRL-Cayley prior, naturally prevents the inferred system from diverging when extrapolated beyond the observed data. In contrast, the trajectories obtained under the free parametrization begin to diverge from the true mean shortly after the inference interval. Although we mitigate the instabilities by assigning a log-likelihood of \(-\infty\) to any sample that yields an unstable trajectory during the MCMC sampling, that does not ensure the posterior draws of the matrices of the SSM in \eqref{eq:SSM} lead to an externally stable system. A similar phenomenon is observed in Fig.\ref{fig:var_plot}, which shows the variance of the output trajectories computed in the same way as in Fig.\ref{fig:avg_plot}. Under the \(Q\)-WNS-BRL-Cayley prior, the posterior trajectory samples remain well-behaved, producing a variance that closely matches the true variance over time. In contrast, the free parametrization exhibits a rapidly increasing variance once we move past the inference interval. These diverging variance trajectories again reflect the fact that we cannot guarantee with a free parametrization approach that the SSM in \eqref{eq:SSM} is externally stable. 

Overall, these figures illustrate that enforcing stability through the \(Q\)-WNS-BRL-Cayley distribution yields posterior {samples} that are robust and exhibit predictive performance that remains consistent even in the extrapolation regime. By contrast, a naive free parametrization of the SSM can cause numerical challenges during MCMC sampling and lead to unreliable predictions, particularly for longer-term forecasts.

\section{Conclusion and Future Work}
In this work, we have introduced a direct parametrization for continuous-time stochastic SSMs that guarantees external stability via the stochastic BRL. This framework not only enables the construction of probabilistic priors that enforce stability almost surely but also facilitates robust Bayesian inference, as demonstrated by our simulation study.

Future work will explore incorporating alternative stability notions and extending the approach to linear parameter-varying systems. Another direction is to experiment with and develop more efficient Bayesian inference methods for high-dimensional SSMs, for example, by leveraging variational inference techniques.